\documentclass[journal,onecolumn,draftcls]{IEEEtran}
\ifCLASSINFOpdf
  \usepackage[pdftex]{graphicx}
\else
  \usepackage[dvips]{graphicx}
\fi
\usepackage[cmex10]{amsmath}
\usepackage{amsfonts,amssymb}
\usepackage{amsthm}
\usepackage{authblk}
\newtheorem{lemma}{Lemma}
\newtheorem{corollary}{Corollary}

\newtheorem{theorem}{Theorem}
\usepackage{algorithm}
\usepackage{algorithmic}

\begin{document}
%
\title{Block Compressed Sensing Based Distributed Device Detection for M2M Communications}
%
%
%

\author[$^{\dagger}$]{Yunyan Chang}
\author[$^{\dagger \dagger}$]{Peter Jung}
\author[$^{\dagger}$]{Chan Zhou}
\author[$^{\dagger \dagger}$]{S\l awomir Sta\'nczak} 
\affil[$^{\dagger}$]{Huawei European Research Center, Riesstr. 25 C, 80992 Munich, Germany \authorcr
    Email: \{yunyanchang, chan.zhou\}@huawei.com}
\affil[$^{\dagger \dagger}$]{Fraunhofer Heinrich-Hertz-Institute and Technische Universit\"at Berlin, Berlin, Germany \authorcr
    Email: \{peter.jung, slawomir.stanczak\}@hhi.fraunhofer.de}
\maketitle

\begin{abstract}
In this work\footnote{This work has been presented in parts in \cite{CJZS16} at the 2016 IEEE International Conferences on Acoustics, Speech and Signal Processing (ICASSP 2016).}, we utilize the framework of \emph{compressed sensing} (CS) for distributed device detection and resource allocation in large-scale \emph{machine-to-machine} (M2M) communication networks. The devices deployed in the network are partitioned into clusters according to some pre-defined criteria, e.g., proximity or service type. Moreover, the devices in each cluster are assigned a unique signature of a particular design that can be used to indicate their active status to the network. 
The proposed scheme in this work mainly consists of two essential
steps: (i) The \emph{base station} (BS) detects the active clusters
and the number of active devices in each cluster using a novel
\emph{block sketching algorithm}, and then assigns a certain amount of resources accordingly. (ii) Each active device detects its ranking among all the active devices in its cluster using an enhanced greedy algorithm and accesses the corresponding resource for transmission based on the ranking.
The particular design for both full-duplex and half-duplex networks are investigated.
By exploiting the correlation in the device behaviors and the sparsity
in the activation pattern of the M2M devices, the device detection
problem is thus tackled as a CS support recovery procedure for a
particular binary block-sparse signal $x\in\mathbb{B}^N$ -- with block sparsity $K_B$ and in-block sparsity $K_I$ over block size $d$. 
Theoretical analysis shows that, using the proposed scheme, 
the activation pattern of the M2M
devices can be reliably reconstructed 
within an acquisition time of 
$\mathcal{O}(\max\{K_B\log N, K_BK_I\log d\})$,  
which achieves a better scaling and 
less computational complexity of $\mathcal{O}(N(K_I^2+\log N))$ compared with standard CS algorithms. Moreover, extensive simulations confirm the robustness of the proposed scheme in the detection process, especially in terms of higher detection probability and reduced access delay when compared with conventional schemes like LTE \emph{random access} (RA) procedure and classic cluster-based access approaches.
\end{abstract}


%
\IEEEpeerreviewmaketitle

\section{Introduction}
\label{sec:intro}
%
%
%
%
Towards the next generation of mobile and wireless networks, \emph{machine-to machine} (M2M) communications \cite{BEH12} is expected to play a significant role and form the basis for the future \emph{Internet of Things} (IoT). M2M communications offers a wide range of applications providing various services \cite{DBK10}, such as the Smart Grid, the E-health system, \emph{Vehicle-to-Vehicle} (V2V) communications, and etc. However, with the rapid growth of the M2M market, the number of M2M devices deployed in a wireless communication network is anticipated to be tremendously increased in the coming years and can potentially exceed 50 billion by the year 2020 \cite{OBHM13}, thus posing significant challenges to current radio access networks.

In LTE and LTE-Advance \cite{LTE11}, once a transmission is triggered by some random event, the device initiates a \emph{random access} (RA) procedure by firstly sending a randomly chosen preamble to indicate its active status, followed by a resource allocation by the serving \emph{base station} (BS) based on the detection \cite{3GPP868}. According to the LTE specifications, each cell is assigned a pool of 64 Zadoff-Chu sequences \cite{3GPP321} as preambles. However, a collision will occur if two or more devices have randomly selected the same preamble. When the number of accessing devices becomes excessively large, the simultaneous access attempts will incur a high probability of collisions in the RA procedure. Hence, the network will easily become overloaded and congested, leading to high detection failure rate and large access delay. 

Apart from a huge number of M2M devices, the messages in M2M communication scenarios are in general highly correlated, e.g., due to proximity, the same service type, and etc \cite{TK12}. As a result, it is reasonable to partition the devices into a number of clusters, thus a balk of similar requests can be handled in a single shot. Moreover, since the M2M traffic is characterized by the sporadic communication among a huge number of devices \cite{DBZ12}, each device has a low probability of being active, thus exhibiting a certain level of sparsity in the device detection process. In this context, \emph{compressed sensing} (CS) \cite{CRT06}-\cite{Don06} is becoming a proper paradigm to deal with high-dimensional signals with a sparse representation, where signal acquisitions can be done in a significantly reduced sampling rate. Moreover, CS has been identified as one of the key drivers behind the efforts to cope with the tremendous scaling issues due to massive connectivity \cite{SJS13}\cite{WBSJ15}.

Motivated by the CS principles and recognizing the cluster-like behavior as well as the sparsity in the M2M traffic, the activation pattern of the M2M devices can be formulated as a particular block sparse signal -- with additional in-block structure -- in CS based applications. Thus the device detection process can be mapped into the recovery procedure of such a sparse signal. In addition, for large-scale networks like M2M communication networks, distributed resource allocation schemes \cite{SSBHU09}, where each device determines its resource allocation autonomously, are in general highly preferred due to their better scalability with the size of the network. Since the centralized control by the BS is not needed, the distributed schemes call for much less coordination and information exchange between the devices, thus the signaling overhead can be substantially reduced and adapted according to the size of the network.

\subsection{Main Contribution}
\label{sec:summary}
In this work, we investigate in a novel distributed device detection scheme of the network activation pattern based on CS techniques to facilitate efficient resource allocation strategies for large-scale M2M communication networks. The proposed scheme targets to cope with tremendous scaling issues aroused by the massive connectivity in M2M communications, especially for enhanced device detection probability and reduced access delay. In the proposed scheme, the M2M devices deployed in the network are partitioned into clusters in advance and devices in each cluster are assigned a unique signature of a particular design for their initial access. After acquisition of the transmitted signatures from the devices, the proposed scheme mainly conducts two essential steps:  

\begin{itemize}
\item (i) The BS detects the active clusters and the number of active devices in each cluster, and then assigns a certain amount of resources accordingly.
\item (ii) Each active device detects its ranking among the active devices in its cluster and accesses the corresponding resource assigned by the BS based on the ranking.
\end{itemize}

A particular signature design for the devices in each cluster is investigated to facilitate the detection process of step (i) at the BS side and step (ii) at the device side. Besides, a novel algorithm based on the Count-Sketch procedure \cite{CCF02}\cite{HB11} is developed for the realization of step (i) to cope with complexity issues due to the massive number of devices. And step (ii) is performed based on some conventional greedy algorithms such as \emph{Orthogonal Matching Pursuit} (OMP) \cite{TG07} and \emph{Iterative Hard Thresholding} (IHT) \cite{BD09} but with side information of the feedback from the BS to enhance the performance robustness and to reduce the computational complexity. 
Theoretical analysis presents a rigorous proof that the proposed scheme is able to efficiently reduce the signal acquisition time with much less computational complexity. In particular, the following theorem holds.

\begin{theorem}
\label{theorem_1}
    Suppose that the activation pattern of the M2M devices is modeled as an arbitrary block sparse binary signal $x\in\mathbb{B}^N$ with block sparsity $K_B$ and in-block sparsity $K_I$ over block size $d$, and the signatures transmitted by the M2M devices are modeled as the measurement matrix $A\in\mathbb{R}^{M\times N}$ following the structure designed in Section \ref{sec:algorithm_matrix}.
		By applying the block sketching algorithm for decoding at the BS and the modified OMP for decoding at the M2M devices, $x$ can be reliably reconstructed by the proposed distributed device detection scheme with overwhelming probability if the dimension of the transmitted signatures
		\begin{equation}
		    M=\mathcal{O}(\max\{K_B\log N, K_BK_I\log d\}),
		\end{equation}
		with computational complexity of $\mathcal{O}(N(K_I^2+\log N))$.
\end{theorem}

For a random $K$-sparse signal of length $N$, $\mathcal{O}(K\log N)$ measurements are sufficient for robust signal recovery by conventional CS algorithms like \emph{Basis Pursuit} (BP) \cite{CR07} or OMP, and the computational complexity is of $\mathcal{O}(N^3)$ for BP \cite{CDS98} and $\mathcal{O}(NK^2)$ for OMP \cite{TG07}. Therefore, the proposed scheme requires much less measurements since $d\ll N$ and significantly reduced computational complexity, thus achieving a better scaling with the increasing network size. 
Moreover, extensive simulations confirm the robustness of the proposed scheme in the detection process, especially in terms of higher detection probability and reduced access delay when compared with conventional schemes like LTE RA procedure and classic cluster-based access approaches.

\subsection{Organization of the Paper}
\label{sec:organization}
The remainder of this paper is organized as follows. Section \ref{sec:background} provides an overview on prior work and known results regarding the target problem. The general introduction of the proposed distributed device detection schemes is presented in Section \ref{sec:scheme}, for both full-duplex and half-duplex networks respectively. Section \ref{sec:model} sets up a mathematical model for the proposed scheme, and the detailed detection algorithms are introduced in Section \ref{sec:algorithms}. Theoretical analysis on the performance guarantees for the proposed scheme are discussed in Section \ref{sec:theory}. In Section \ref{sec:simulation}, numerical results are presented and evaluated. Finally, Section \ref{sec:conclude} concludes the paper with some final remarks.

\subsection{Notational Remarks}
\label{sec:notation}
Throughout this work, uppercase and lowercase boldface letters denote
matrices and vectors, respectively. 
The superscript $(\cdot)^T$ indicates the transpose of a matrix or a vector, 
and the field of binary, real and complex numbers are denoted by $\mathbb{B}$, $\mathbb{R}$ and $\mathbb{C}$. The cardinality of a set is given by $|\cdot|$, and the $\ell_2$-norm is given by $||\cdot||$. Probability is denoted by $\mathbb{P}(\cdot)$ and $\mathbb{E}[\cdot]$ indicates the expectation operator. 
Furthermore, $\mathcal{O}$ denotes ``big-O'' according to Knuth's notation.
Unless otherwise stated, all logarithms are assumed to be to the base 2.

\section{Background}
\label{sec:background}

\subsection{Compressed Sensing}
\emph{Compressed sensing} (CS) has emerged as a new framework for efficient signal
acquisition and reconstruction by finding solutions to underdetermined
linear systems with sparsity priors or related compressibility properties. 
CS builds on the revelation that a signal having a sparse representation in one basis can be recovered from a small number of projections onto another basis that is incoherent with the first \cite{CRT06}-\cite{Don06}. The CS is a concrete protocol for sensing and compressing data simultaneously.

Let $x\in \mathbb{R}^N$ be an original signal to be recovered and
$\Theta\in\mathbb{R}^{N\times N}$ be a representing basis or
dictionary such that the signal can be expressed as
\begin{equation}
    x=\Theta \theta,
\end{equation}
where $\theta\in\mathbb{R}^N$ is the corresponding coefficient vector. If all except for $K$ entries of $\theta$ are zero, i.e. 
$\lVert\theta\rVert_0:=|\{i:\theta_i\neq 0\}|\leq K$, then the signal $x\in\mathbb{R}^N$ is called to be $K$-sparse.
In this case the main objective is to recover $x$ from $M$ compressive
(undersampled, that is $M\ll N$) measurements $y\in\mathbb{R}^M$ given by
\begin{equation}
    y=\Phi x,
\end{equation}
where $\Phi\in\mathbb{R}^{M\times N}$ is usually called the measurement matrix.

Several sufficient conditions on the measurement matrix $\Phi$ are known to ensure that a $K$-sparse signal $x$ can be
accurately reconstructed from the compressive measurements $y$. A
widely considered condition is the \emph{restricted isometry
  property} (RIP) \cite{CT05} stating that there should exist a 
constant $\delta_K\in[0,1)$ such that 
$(1-\delta_K)||x||^2_2 \leq ||\Phi x||^2_2 \leq (1+\delta_K)||x||^2_2$
holds for all $K$--sparse vectors $x$. RIP ensures that $K$-sparse
vectors are mapped almost--isometrically by $\Phi$. 
However, it is a well-known fact that RIP is only sufficient but not
necessary for uniform recovery of sparse signals. Overcoming this
issue, the \emph{nullspace property} (NSP) provides both
sufficient and necessary conditions on $\Phi$ for robust sparse signal
recovery, which basically requires that any vector in the kernel space
of $\Phi$ is far from being sparse (see, e.g., the book \cite{FR13}
for details and additional references on RIP and NSP). Random matrices 
with \emph{independent and identically distributed} (i.i.d.)
(sub--) Gaussian \cite{CD06} and $\pm1$--Bernoulli \cite{CT06} distributed
entries follow the RIP and NSP condition \emph{with high probability}
(w.h.p.) in the $(K,M,N)$--regime:
\begin{equation}
    M\geq cK\log (N/K),
\end{equation}
where $c$ is some constant (see exemplary \cite{CW08} for an overview on these results).

Several algorithms have been proposed for reconstruction of the original signal $x$ from the undersampled measurements $y$. The canonical approach \cite{CDS98}\cite{CT04} uses linear programming to solve the $\ell_1$ minimization problem
\begin{equation}
    \hat{\theta}=\arg \min_{\theta} ||\theta||_1 \quad \text{subject to} \quad y=\Phi x=\Phi\Theta\theta.
\end{equation}

This problem can be solved in polynomial time but the computational
complexity of $\mathcal{O}(N^3)$ \cite{CDS98} renders it impractical
for many applications. Additional methods have been proposed involving
greedy pursuit methods such as OMP \cite{TG07} and IHT \cite{BD09},
which build up a signal approximation iteratively by making locally
optimal decisions, and bring the computational complexity down to the
order of $\mathcal{O}(K^2N)$.  In some practical applications 
even this reduced complexity can be a bottleneck as it is the case in our
setting and then fast sketching algorithms need to be considered.

\subsection{Distributed Resource Allocation}
In large-scale networks, it is preferred that the algorithms for resource allocation are amenable to distributed implementation \cite{SSBHU09} .
Distributed resource allocation schemes allow each transmitter to determine its local decisions autonomously without the need for global coordination, which introduces an inherent robustness to communication failures and environmental uncertainties. Therefore, distributed schemes call for significantly less information exchange between wireless devices. As a result, the signaling overhead can be substantially reduced and adapted according to the network size. Therefore, distributed schemes in general achieve better scalability with the number of devices in a network, which is an attractive property for large-scale networks.

Network resource allocation is often modeled as a utility maximization problem \cite{PC06}-\cite{SWB09}, which is flexible to accommodate a wide range of performance metrics through an appropriate assignment of utility functions. 
A number of gradient/subgradient methods \cite{TPC06} have been proposed to solve such problems in a distributed way. However, these methods are sensitive to the choice of stepsizes, leading to slow and unstable convergence. In \cite{SWB09}, some distributed Newton-like algorithms \cite{BJKJ15} are introduced that obtain high tolerance of errors in the stepsize calculation. There also exist faster numerical convex problem solvers, e.g., the nonlinear Jacobi algorithm \cite{BT89}, which converges fast under certain contraction conditions. Furthermore, pricing algorithms in distributed systems have also been widely approached with game theory \cite{JT04}\cite{MMHBR09}, since those scenarios fit naturally within the game theoretic frameworks.

\subsection{Related Works}
There are many approaches to the problem of device detection in large-scale M2M communication networks based on the CS theory. 

Apart from the LTE RA procedure \cite{LTE11}, \cite{3GPP321}, \cite{CEA11} and \cite{JPL10} proposed some cluster-based approaches to manage the M2M traffic in the uplink direction. According to these approaches, wireless devices are partitioned into clusters and a cluster head is selected for each cluster to handle the uplink traffic from the cluster members. The cluster heads process and fusion the requests from their cluster members, and then forward the accumulated messages periodically to the serving BS. However, this approach calls for efficient interaction between the cluster head and the cluster members, which leads to large signaling overhead and excessive processing delay.

In \cite{MLH09}, Meng \emph{et.al} simply adopted the detection algorithms from the probabilistic Bayesian framework for the sparse event detection in \emph{wireless sensor networks} (WSN). In \cite{SD11}, the authors applied greedy CS detection algorithms for jointly decoding the multi-user activity and data in \emph{Code Division Multiple Access} (CDMA) systems. It is shown that an application of the CS methods can greatly reduce the length of the spreading sequences, thereby leading to a higher network throughput. 
However, both of the schemes directly applied the standard CS decoding algorithms for device detection at the central controller, which ignores the correlation in the behavior among the M2M devices and the additional sparsity in the activation pattern, thus leading to non-optimality in the processing especially in terms of efficiency and computational complexity.

The authors of \cite{LT10} proposed a sparse signal recovery scheme via decentralized in-network processing for event detection in WDN based on a consensus optimization formulation. However, since a random sleeping strategy is used to enforce the compressive data collection, the detection field for each active time slot is rather uncertain and may lead to incorrect detection decisions with high probability. Moreover, none of the schemes exploits the cluster-like behavior among the devices, which exhibits a certain level of correlation in the device activations.


\section{Proposed Scheme Description}
\label{sec:scheme}

Standard CS decoding algorithms need to be appropriately adjusted in order to facilitate the specific sparsity feature in the activation patterns of the M2M devices for more efficiency in the device detection and resource allocation process. Furthermore, the detection scheme should be implemented with low computational complexity and distributed processing possibility, since the M2M devices are ordinarily configured with limited processing capability and battery power. In this section, we introduce an efficient mechanism for distributed device detection and resource allocation for large-scale M2M communication networks, which fully exploits the specific sparsity feature in the activation pattern of the M2M devices.

Since the M2M devices in general exhibit highly correlated behaviors, e.g., due to proximity, the same service type, and etc, it is reasonable to partition the devices into a number of clusters, thus a balk of similar requests can be handled in a single shot. And a device can be assorted into multiple clusters based on different features. The cluster structures are known both at the BS and at the devices during the device registration process to the network and updated periodically.
Besides, both the BS and each device in the network have the knowledge of channel information from other devices to itself, which can be possibly obtained via statistical channel knowledge \cite{ZZJZG15}, location-based estimation \cite{WQDCSC15}, channel reciprocity \cite{KJGK10} or long-term observation \cite{SRBF01}.

Furthermore, a ranking among the devices in each cluster is conducted in advance in order to determine the order for the active devices to access a pool of assigned resource blocks by the BS. As a result, the active devices access the corresponding resources according to their ranking in the cluster. The devices can be ranked according to some pre-defined rule, e.g., based on the order of the device ID or the service priority. And the ranking information is also informed to all the M2M devices in advance by the BS after the device registration.

In addition, since the M2M traffic exhibits a certain level of sparsity in the detection activity, each device has a low probability of being active at a certain time slot. In general, we observe that at a given time instant, only several clusters become active and only a small number of devices in those active clusters are triggered to access the network, where a cluster is called to be ``active'' if one or more devices from the cluster are active. Thus, a twofold sparsity pattern, namely the \emph{block sparsity} and \emph{in-block sparsity}, can be defined to model the active status of the devices.
\begin{itemize}
\item \textbf{Block Sparsity:} Only several of the cluster become active at a certain time instant.
\item \textbf{In-block Sparsity:} Only several devices from the same cluster are active.
\end{itemize} 

In our proposed scheme, the devices in each of the clusters are assigned a unique signature used to indicate their active status to the network, the design of which is different from the preambles (Zadoff-Chu sequences) in the LTE RA procedure and will be introduced in details in Section \ref{sec:algorithm_matrix}. Each of the signatures indicates the membership to a particular cluster, and is informed to the devices in the corresponding clusters in advance.

A full-duplex system allows data to be transmitted in both communication directions at the same time, which means the signals can be transmitted and received simultaneously, while half-duplex systems only allow communication in a single direction at a time. Usually the terms ``full-duplex'' and ``half-duplex'' are used to describe point-to-point systems, whereas in this work we extend them to multi-point systems composed of multiple connected entities that communicate with each other. Therefore in this work, ``full-duplex'' means an M2M device can both transmit and receive data on the wireless channel simultaneously. Depending on the network configuration whether the full-duplex mode is supported by the devices or not, the proposed scheme for distributed device detection and resource allocation in large-scale M2M networks is differentiated into the following two cases.

\subsection{Full-Duplex M2M Network}
\label{sec:full-duplex}
If the communication network is configured to support the full-duplex mode, the proposed scheme mainly consists of four phases, which are illustrated in Figure \ref{fig:phases}.

    \begin{figure}[htb]
				\centering
					\includegraphics[width=5.0in]{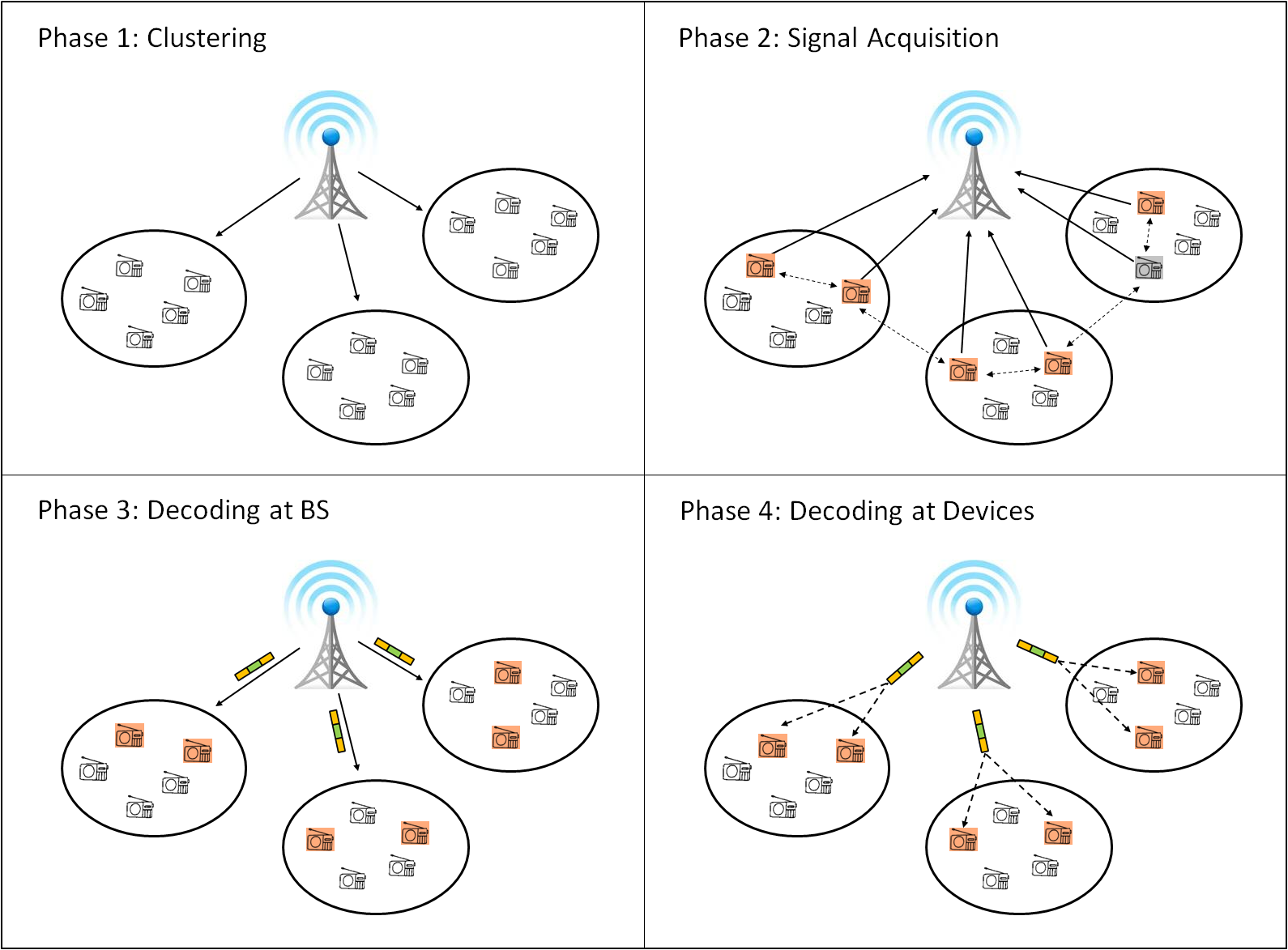}
        \caption{Major steps for the proposed scheme in full-duplex mode}
        \label{fig:phases}
    \end{figure}

\begin{itemize}
\item \textbf{Phase 1a: Clustering} As the devices are partitioned into clusters and the devices in each cluster is designed to obtain a unique signature, all these relevant information is informed to the devices and updated periodically by the BS. The information includes the cluster IDs, the devices in each cluster, the ranking information of the devices, and the uniquely-designed signatures distributed to each of the clusters.

\item \textbf{Phase 2a: Signal Acquisition} Once the devices become active and start accessing the network, they transmit simultaneously the assigned signatures to indicate their active status to the network. With full-duplex transceivers, all the devices and the BS receive individual linear combinations of the transmitted signatures.

\item \textbf{Phase 3a: Decoding at BS} The BS detects the active clusters, the number of active devices in each cluster, as well as the collision patterns in the received measurements. Then it broadcasts this information to the devices and assigns a certain amount of resources to each of the detected active clusters accordingly.

\item \textbf{Phase 4a: Decoding at Devices} Each active device performs device detection for its corresponding cluster using its received signal and the broadcast information from the BS in Phase 3. Then it detects its ranking among all the active devices in the cluster and accesses the corresponding resource block assigned by the BS for transmission. 
\end{itemize}

\subsection{Half-Duplex M2M Network}
\label{sec:half-duplex}
If the communication network is configured as a half-duplex system, then we introduce a cluster head for each of the clusters in advance. The cluster heads are selected with the property not to transmit concurrently with the other devices in the same cluster. Besides, the cluster heads are supposed to know the cluster structures as well as the estimated channel information from the cluster members to themselves. To this end, the proposed scheme under the half-duplex assumption performs the procedures as follows.

\begin{itemize}
\item \textbf{Phase 1b: Clustering} The cluster information including the cluster heads and the information involved in Phase 1a under the full-duplex assumption is informed to the devices by the BS during the initial device registration to the network. 

\item \textbf{Phase 2b: Signal Acquisition} After the active devices transmit their individual signatures to indicate their active status, both the BS and the cluster heads collect their own measurements, which are linear combinations of the transmitted signatures.

\item \textbf{Phase 3b: Decoding at BS} Same as Phase 3a under the full-duplex assumption.

\item \textbf{Phase 4b: Decoding at Cluster Heads} The cluster heads of the active clusters perform device detection for their corresponding clusters and detect the ranking of the active devices in their clusters using the same decoding algorithms as proposed for Phase 4a under the full-duplex assumption.

\item \textbf{Phase 5b: Resource Access} Each cluster head broadcasts the detected ordering information of the active devices to the members in the cluster, and those active devices access the corresponding resource blocks for transmission based on their ranking afterwards.
\end{itemize}


The message flows involved in the proposed schemes under both full-duplex and half-duplex assumptions are illustrated in Figure \ref{fig:flow}.

    \begin{figure}[htb]
				\centering
					\includegraphics[width=5.0in]{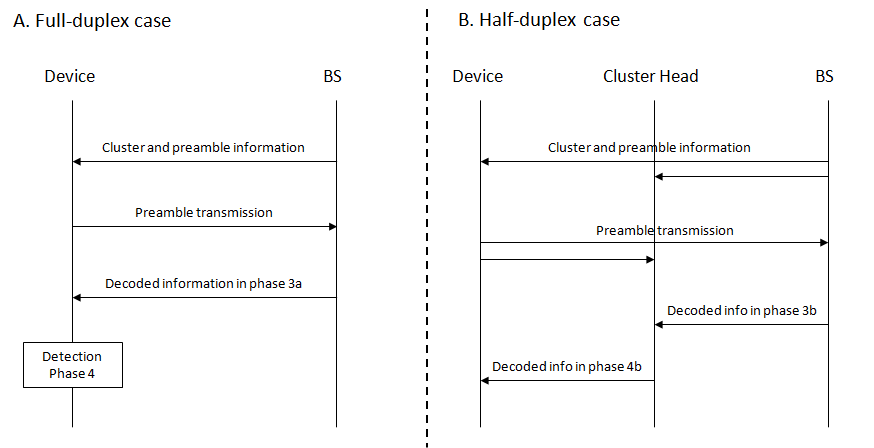}
        \caption{Message flow charts under both full-duplex and half-duplex assumption}
        \label{fig:flow}
    \end{figure}

\section{System Model}
\label{sec:model}
In this section, we set up the mathematical model to describe the target problem and elaborate on the proposed scheme introduced in Section \ref{sec:scheme}. 

\subsection{Transmitter Side}
\label{sec:tx_side}
Consider an M2M network with $N$ devices, which are partitioned in advance into $L$ clusters of equal size $d$ according to some pre-defined criteria. As defined in Section \ref{sec:scheme}, the twofold sparsity pattern of the M2M devices, namely the \emph{block-sparsity} and \emph{in-block-sparsity}, are $K_B$ and $K_I$, respectively. That is, only $K_B$ out of $L$ clusters are active, and the number of active devices in each cluster is at most $K_I$. 
In addition, we denote $\mathcal{S}_B$ as the \emph{block support}, which is defined to be the set of indexes of the active clusters. Similarly, the \emph{in-block support}, denoted as $\mathcal{S}_{I,\ell}$, indicates the set of indices of the active devices in cluster $\ell$. Since the activation pattern of the M2M devices is $K_B$ block sparse and $K_I$ in-block sparse, we have $|\mathcal{S}_B|\leq K_B$ and $|\mathcal{S}_{I,\ell}|\leq K_I$ for all $\ell\in\{1,\cdots,L\}$.Therefore, the total number of active devices in the network is $K\leq K_BK_I$. Due to the sparse nature of the event occurrence in MTC, we have $K\leq K_BK_I\ll N=Ld$. 

Herein, we denote a $K$-sparse binary sequence $x\in\mathbb{B}^N$ to model the activation pattern of the M2M devices in the network, with entry ``1'' indicating the corresponding device to be active and ``0'' otherwise. Furthermore, we denote $x_{\ell}\in\mathbb{B}^d,\ell\in\{1,\cdots,L\}$ as the status vector for cluster $\ell$. Thus, the activation pattern of the devices is formulated as a particular block sparse signal -- with additional in-block structure -- in the CS based applications.

We apply the CS theory to the transmission incurred by the M2M devices in the network. To this end, let $A\in \mathbb{R}^{M\times N}$ be the measurement matrix whose exact structure is defined later in Section \ref{sec:algorithm_matrix}. Each column of $A$, say column $i$ denoted by $a_i, i\in\{1,\cdots,N\}$, corresponds to the unique signature sent by the M2M device indexed by $i$ if it is active. Besides, we denote $A_{-,\ell}\in \mathbb{R}^{M\times d}$ as a submatrix of $A$ corresponding to the signatures sent by the devices from the $\ell$-th cluster.

\subsection{Receiver Side}
\label{sec:rv_side}
The signatures transmitted by the active devices get superimposed in the wireless channel, and the signal $y\in\mathbb{C}^M$ received by the BS is a linear combination of the transmitted signatures, which is given by
\begin{equation}
    y=AH_Bx+\epsilon,
\label{eq_1}
\end{equation}
where $H_B\in\mathbb{C}^{N\times N}$ is the diagonal channel matrix, and $\epsilon\in\mathbb{C}^M$ is the thermal noise vector having random, zero-mean components of variance $\sigma^2$. 

In the meantime, each active device (under full-duplex assumption) or the cluster head (under half-duplex assumption) also collects a linear combination of the transmitted signatures,where the received signal $y_D\in\mathbb{C}^M$ at a particular device is given by
\begin{equation}
    y_D=AH_Dx+\epsilon_D,
\label{eq_2}
\end{equation}
where $H_D$ is an ${N\times N}$ diagonal matrix representing the wireless channels from the M2M devices to this device, and $\epsilon_D\in\mathbb{C}^M$ is the noise vector with zero mean and variance of $\sigma^2$. 

The cluster structures are assumed to be known both at the BS side and at all the M2M devices. Furthermore, all the nodes in the network have the channel information from other nodes to itself as well as to the BS, which can be possibly obtained via statistical channel knowledge, location-based estimation, channel reciprocity or long-term observation. With these knowledge at hand, we apply the CS related techniques in the proposed distributed device detection and resource allocation scheme to reconstruct the $K$-sparse signal $x$ from the received signal $y$ at the BS and $y_D$ at the devices.

\subsection{Problem Formulation}
\label{sec:problem}
    \begin{figure}[htb]
				\centering
					\includegraphics[width=5.5in]{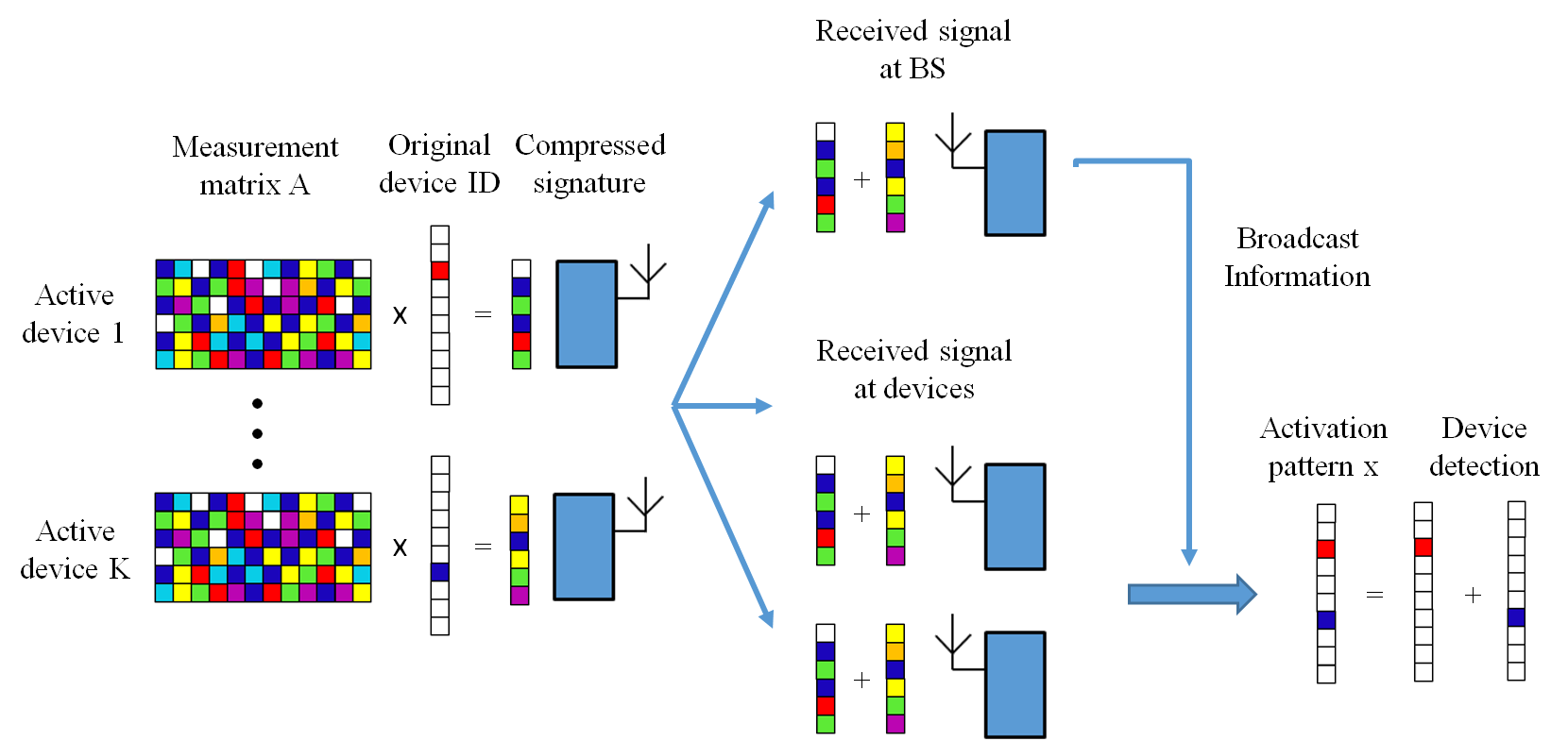}
        \caption{Target problem formulation}
        \label{fig:problem}
    \end{figure}

As discussed in Section \ref{sec:scheme}, the transmission scenario and the target problem can be illustrated as in Figure \ref{fig:problem}.
First, each of the active devices transmits a specially designed signature to the network to indicate its active status, which corresponds to a certain column of the measurement matrix $A$. The signatures transmitted by the active devices get superimposed in the wireless channel, and both the BS and the devices collect a linear combination of the transmitted signatures. 
On one hand, the BS has to detect the active clusters as well as the number of active devices in each cluster, then it broadcasts this information to the devices and assigns a certain amount of resource blocks to the detected clusters accordingly; and on the other hand, each active device needs to perform the device detection in a distributed manner for its corresponding cluster with its own collected measurements and the feedback information from the BS, and then accesses the assigned resource blocks by the BS based on its detected ranking among all the active devices in the cluster.

Herein, mapping to the mathematical model, the object of interest would be at first to perform the \emph{block support recovery} at the BS in order to obtain an accurate estimate of $\mathcal{S}_B$ and $|\mathcal{S}_{I,\ell}|$ for all $\ell\in\{1,\cdots,L\}$, and thereafter, to perform the \emph{in-block support recovery} at each active device from cluster $\ell$ (under full-duplex assumption) or at the cluster head (under half-duplex assumption) in order to obtain an accurate estimate of the in-block support $\mathcal{S}_{I,\ell}$ . 

To be specific, the challenges we need to tackle in this work can be summarized as:
\begin{itemize}
\item \emph{Challenge 1:} How to obtain an accurate estimate of $\mathcal{S}_B$ and $|\mathcal{S}_{I,\ell}|$ for all $\ell\in\{1,\cdots,L\}$ at the BS?
\item \emph{Challenge 2:} How to obtain an accurate estimate of $\mathcal{S}_{I,\ell}$ at the device side for cluster $\ell$?
\end{itemize}

\section{Algorithm Design}
\label{sec:algorithms}
In this section, we introduce in details about the proposed algorithms in order to tackle the two challenges posed in Section \ref{sec:problem}, particularly for the structured signature model, and the decoding procedure at the BS as well as at the M2M devices, respectively.

\subsection{Structured Random Signature Model}
\label{sec:algorithm_matrix}
The measurement matrix $A\in \mathbb{R}^{M\times N}$ that we design here is a structured random matrix, which is an extension of those utilized by the Count-Sketch procedure proposed in \cite{CCF02}\cite{HB11} due to its preferable advantage of low computational complexity. We denote by $\mathcal{A}(R,T,L,d,\alpha)$ a particular distribution over matrices having $RT$ rows and $Ld$ columns, and we assume that the measurement matrix $A$ is drawn from this distribution, i.e., $A\sim\mathcal{A}(R,T,L,d,\alpha)$ with $M=RT, N=Ld$.
		
As illustrated in Figure \ref{fig:matrix}, the measurement matrix $A$ is composed of the vertical concatenation of $T$ individual random matrices, denoted $A_{t,-}\in\mathbb{R}^{R\times N}$ for $t\in\{1,\cdots,T\}$. Meanwhile, each $A_{t,-}$ consists the horizontal concatenation of $L$ sub-matrices $A_{t,\ell}\in\mathbb{R}^{R\times d}$ for $\ell\in\{1,\cdots,L\}$. Each $A_{t,\ell}$ is a sparse matrix containing exactly $d$ non-zero components - located on the same row and with the same value. The index of the row with non-zero elements is chosen uniformly at random from the set $\{1,2,\cdots,R\}$, and the non-zero component takes the value of $\pm\alpha$ with probability $1/2$. For a given realization of $A_{t,\ell}$, let $q_{t,\ell}\in\{1,\cdots,R\}$ denote the index of the row of $A_{t,\ell}$ with non-zero entries, and $s_{t,\ell}\in\{-\alpha,+\alpha\}$ be the corresponding value of the non-zero components in $A_{t,\ell}$.

To this end, each of the signatures transmitted by the M2M devices, which is the corresponding column of the structured measurement matrix $A$, is a sparse sequence of length $M$ with sparsity level $T$.

    \begin{figure}[htb]
				\centering
					\includegraphics[width=5.6in]{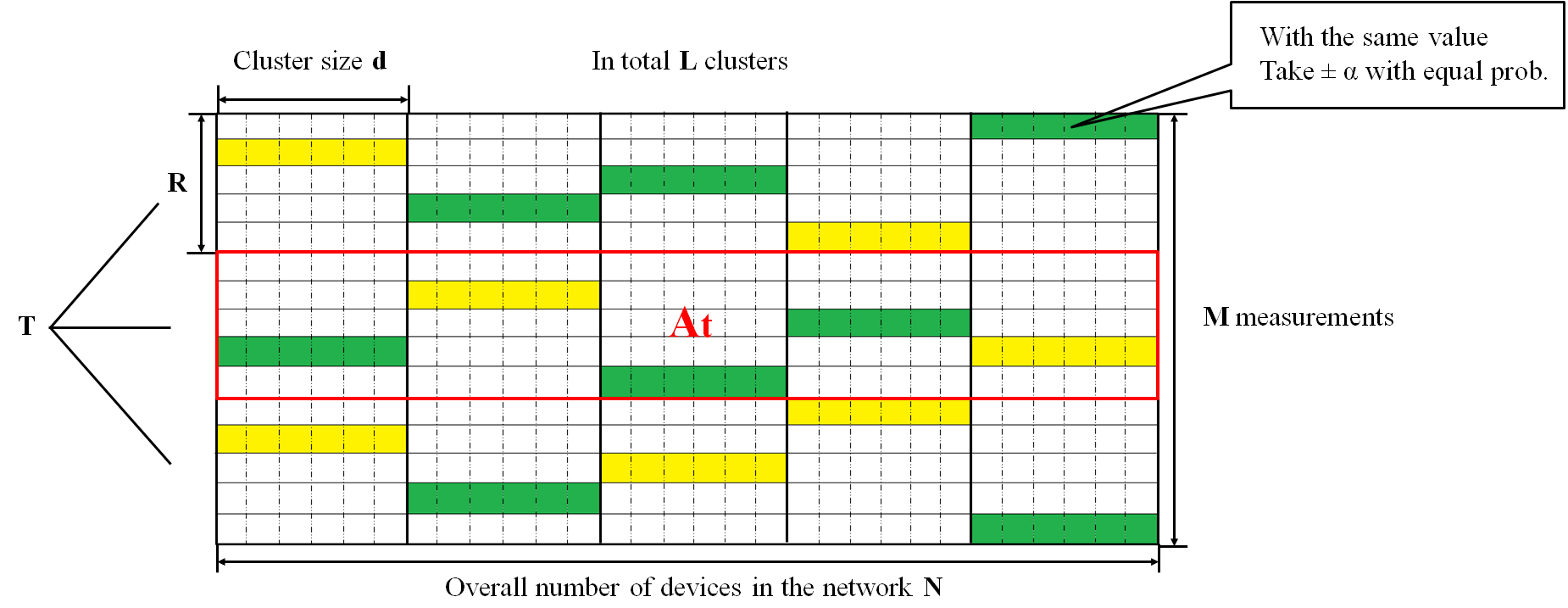}
        \caption{The structure of the measurement matrix $A$}
        \label{fig:matrix}
    \end{figure}

\subsection{Block Support Recovery at BS}
\label{sec:algorithm_BS}
The decoding procedure at the BS side aims to tackle the \emph{Challenge 1} raised in Section \ref{sec:problem}, which is to obtain an accurate estimation of the block support $\mathcal{S}_B$ and the cardinality of the in-block support $|\mathcal{S}_{I,\ell}|$ for all each active cluster $\ell\in\{1,\cdots,L\}$.

The signal $y$ received by the BS at some given time instant is given by (\ref{eq_1}). 
Since the channel knowledge $H_B$ is assumed to be known at the M2M
devices, some pre-channel correction can be performed before the
corresponding signatures are transmitted by the active
devices. Particular variants of (generalized) inverses of 
the channel matrix may be taken at the transmitter side. 
For example, for weak links with too small channel magnitudes due to the
near-far behavior it is more reasonable for the device to stay offline
and avoid excessively large transmit power and strong interference to
other nodes. In this case $H_B$ is always invertible since the
smallest entry is perpetually above certain threshold
and then the obtained measurements at the BS are given as
\begin{equation}
    y = AH_BH_B^{-1}x + \epsilon = Ax+\epsilon.
\label{equ_y}
\end{equation}
A different approach for this detection phase is presented by one of the authors in
\cite{KJ16}. There only channel phases and not the amplitudes of the
individual channels are corrected and, interestingly, the resulting 
\emph{nonnegative and sparse} recovery problem can be solved also
efficiently using nonnegative least--squares without sparsity--promoting regularizers.
Here, we develop a fast block sketching algorithm based on the Count-Sketch procedure proposed in
\cite{CCF02}\cite{HB11} to realize the decoding process at the BS,
which is implemented as follows. We denote $y_t\in\mathbb{R}^R$ as the subvector of $y$ corresponding to observations obtained via the submatrix $A_{t,-}$, such that
\begin{equation}
    y_t=A_{t,-}x+\epsilon_t, \quad \text{for} \quad t\in\{1,\cdots,T\}.
\end{equation}

Then for each $t\in\{1,\cdots,T\}$, we form the signal estimates $\tilde{x}_t\in\mathbb{R}^N$ by indexing and scaling the entries of the corresponding observations $y_t$, where 
\begin{equation}
    \tilde{x}_t=A_{t,-}^Ty_t, \quad \text{for} \quad t\in\{1,\cdots,T\}.
\end{equation}

Recall that each $A_{t,-}$ consists the horizontal concatenation of $L$ submatrices $A_{t,\ell}$ for $\ell\in\{1,\cdots,L\}$, which are sparse matrices containing $d$ non-zero components located on the same row indexed by $q_{t,\ell}\in\{1,\cdots,R\}$ and with the same value $s_{t,\ell}\in\{-\alpha, +\alpha\}$. As a result, an individual entry of $\tilde{x}_t$ indexed by $i$ which belongs to the $\ell$-th block can be expressed by 
\begin{equation}
    \tilde{x}_{t,i}=s_{t,\ell}y_{t,\ell}, \quad \text{for} \quad i\in\{d\ell-d+1,\cdots,d\ell\}, \ell\in\{1,\cdots,L\},
\label{eq:theory_x0}
\end{equation}
where $y_{t,\ell}$ is the corresponding entry of $y_{t}$ indexed by $q_{t,\ell}$, which is given as
\begin{equation}
  \begin{aligned}
    y_{t,\ell}&=\sum_{\substack{k\in\mathcal{S}_{t,\ell}\\ \mathcal{S}_{t,\ell}=\{j=\{1,\cdots,L\}:q_{t,j}=q_{t,\ell}\}}}\sum_{i=dk-d+1}^{dk}s_{t,k}x_i+\epsilon_{t,\ell}\\
		&=\sum_{i=d\ell-d+1}^{d\ell}s_{t,\ell}x_i+\sum_{k\in\mathcal{S}_{t,\ell}\backslash\{\ell\}}\sum_{i=dk-d+1}^{dk}s_{t,k}x_i+\epsilon_{t,\ell}\\
		&\overset{(a)}{=}s_{t,\ell}\left(\sum_{i=d\ell-d+1}^{d\ell}x_i\right)+\sum_{k\in\mathcal{S}_{t,\ell}\backslash\{\ell\}}s_{t,k}\left(\sum_{i=dk-d+1}^{dk}x_i\right)+\epsilon_{t,\ell}\\
		&\overset{(b)}{=}s_{t,\ell}|\mathcal{S}_{I,\ell}|+\sum_{k\in\mathcal{S}_{t,\ell}\backslash\{\ell\}}s_{t,k}|\mathcal{S}_{I,k}|+\epsilon_{t,\ell}\\
		&\overset{(c)}{=}s_{t,\ell}|\mathcal{S}_{I,\ell}|+\sum_{k\in\mathcal{S}_B\cap\mathcal{S}_{t,\ell}\backslash\{\ell\}}s_{t,k}|\mathcal{S}_{I,k}|+\epsilon_{t,\ell},
	\end{aligned}
\label{eq:theory_y1}
\end{equation}
where (a) follows from the structure of $A_{t,\ell}$ with equal non-zero elements on the same row, (b) follows since $x_i\in\{0,1\}$ is drawn from a binary ensemble, and (c) follows due to the cardinality of the in-block support $|\mathcal{S}_{I,k}|=0$ for non-active blocks.

Taking $y_{t,\ell}$ in (\ref{eq:theory_y1}) back to (\ref{eq:theory_x0}), we can get
\begin{equation}
    \begin{aligned}
		    \tilde{x}_{t,i}&=s_{t,\ell}\left(s_{t,\ell}|\mathcal{S}_{I,\ell}|+\sum_{k\in\mathcal{S}_B\cap\mathcal{S}_{t,\ell}\backslash\{\ell\}}s_{t,k}|\mathcal{S}_{I,k}|+\epsilon_{t,\ell}\right)\\
				&=\alpha^2|\mathcal{S}_{I,\ell}|+\underbrace{\sum_{k\in\mathcal{S}_B\cap\mathcal{S}_{t,\ell}\backslash\{\ell\}}s_{t,\ell}s_{t,k}|\mathcal{S}_{I,k}|}_{\Delta_{t,\ell}}+\underbrace{s_{t,\ell}\epsilon_{t,\ell}}_{\epsilon_{t,\ell}},\\
		\end{aligned}
\label{eq:theory_x1}
\end{equation}
where $\Delta_{t,\ell}$ is the interference factor from other active blocks, and $\epsilon_{t,\ell}$ is the noise factor with zero mean and variance $\gamma^2$.

Thereafter, in order to mitigate the interference from other blocks on the estimates, we form a block-wise estimate $\bar{x}_{\ell}$ for each $\ell\in\{1,2,\cdots,L\}$, which is obtained as
\begin{equation}
     \bar{x}_{\ell} = \text{median} \{\tilde{x}_{t,i}\}_{t=1,i=d\ell-d+1}^{T,d\ell}, \quad \text{for} \quad \ell\in\{1,2,\cdots,L\}.
\label{eq:theory_x3}
\end{equation}

In other words, the estimate $\bar{x}_{\ell}$ for each block $\ell\in\{1,\cdots,L\}$ is obtained as the median of $\tilde{x}_{t,i}$ over $\mathcal{O}(Td)$ samples. The rational for forming the signal estimate using the median instead of the mean is mainly due to its robustness against outliers, where large value elements in the data stream may spoil some subsets of the estimate during the calculation of the mean.

It will be proven later in Section \ref{sec:theory_BS} that by taking the median value block-wisely among all individual estimations as in (\ref{eq:theory_x3}), each of the estimates $\bar{x}_{\ell}$ for the $\ell$-th block corresponds to $|{\mathcal{S}}_{I,\ell}|$ with overwhelming probability. To this end, the ultimate goal is achieved for block support recovery at the BS. Furthermore, the cardinality of the in-block support set $|{\mathcal{S}}_{I,\ell}|$ can be obtained as
\begin{equation}
    |{\mathcal{S}}_{I,\ell}|=\left[\frac{1}{{\alpha}^2}\cdot\bar{x}_{\ell}\right], \quad \text{for} \quad \ell\in\{1,2,\cdots,L\}.
\end{equation}

Therefore, since $|{\mathcal{S}}_{I,\ell}|$ indicates the number of active devices in cluster $\ell\in\{1,2,\cdots,L\}$, those clusters with $|{\mathcal{S}}_{I,\ell}|>0$ are marked as ``active'' and detected by the BS.

Besides, for a given $x_i$ from an active block $\ell\in\mathcal{S}_B$, if an individual estimate $\tilde{x}_{t,i}$ in (\ref{eq:theory_x0}) is much larger than the block-wise estimate $\bar{x}_{\ell}$ in (\ref{eq:theory_x3}), i.e., $\tilde{x}_{t,i}\gg\bar{x}_{\ell}$, it indicates that the corresponding measurement, i.e., the entry of $y_t$ indexed by $q_{t,\ell}$, suffers strong interference from the other active clusters. That is, for a given $x_i$ from block $\ell$ and a particular $t\in\{1,\cdots,T\}$, the condition for the interference factor $\Delta_{t,\ell}$ in (\ref{eq:theory_x1}) to be 0 doesn't hold. Thus we mark the measurement as ``collided'' for cluster $\ell$ under this case, and keep its index in the collision pattern vector $Q_{\ell}$ for the corresponding cluster.

To this end, \emph{Challenge 1} has been exquisitely resolved, where the accurate estimations are obtained by the BS including the active cluster set $\mathcal{S}_B$, the number of active devices in each active cluster $|\mathcal{S}_{I,\ell}|$, and the collision patterns $Q_{\ell}$ in the measurements. Then it broadcasts these information to the M2M devices and assigns a certain amount of resources to the active clusters based on the detection.

\subsection{In-block Support Recovery at Devices}
\label{sec:algorithm_device}
During the signal acquisition phase, either the active devices (under full-duplex assumption) or the cluster heads (under half-duplex assumption) also collect their own measurements, which are linear combinations of the transmitted signatures from the active devices. In this section, we aim to tackle the \emph{Challenge 2} raised in Section \ref{sec:problem} with the obtained measurements at the device side and the feedback from the BS as side information, in order to obtain an accurate estimation of the in-block support $\mathcal{S}_{I,\ell}$ for each active cluster $\ell\in\{1,\cdots,L\}$.

Taking the measurement matrix $A$ under the random structured model in Section \ref{sec:algorithm_matrix} and the pre-channel correction in (\ref{equ_y}), 
the measurements $y_D$ collected at an active device or at a cluster head are given by
\begin{equation}
    y_D = AH_DH_B^{-1}x + \epsilon_D = \tilde{A}x+\epsilon_D.
\label{equ_yD}
\end{equation}

Since the channels between the devices and to the BS are assumed to be i.i.d. due to large-scale networks, the matrix $\tilde{A}$ also turns to have random i.i.d entries.
Furthermore, advanced technologies such as frequency hopping \cite{KMW01} can be applied to further increase the randomness in $\tilde{A}$.
According to the specific structure of the measurement matrix $A$, for a given cluster $\ell\in\{1,\cdots,L\}$, the corresponding submatrix $A_{-,\ell}$ has only $T$ rows with non-zero components, whose index are denoted by the set $D_{\ell}$. Thus, in order to perform the in-block support recovery of $x_{\ell}$ at an M2M device from cluster $\ell$, we simply need to focus on those $T$ measurements composed of the entries of $y_D$ corresponding to $D_{\ell}$. Furthermore, with the feedback information from the BS on the collision pattern in the measurements $Q_{\ell}$ for cluster $\ell$, those collided measurements detected by the BS, which suffer strong interference from the other clusters, can be discarded for more reliable processing. Hence, we form an index set $U_{\ell}=D_{\ell}\cap\bar{Q}_{\ell}$, and focus merely on $y_{D,\ell}$ -- a vector composed of the entries of $y_D$ corresponding to $U_{\ell}$. Denote $\tilde{A}_{D,\ell}$ as a $|U_{\ell}|\times d$ submatrix of $\tilde{A}$ with vertical concatenation of rows corresponding to $U_{\ell}$ and columns for block $\ell$. Therefore, we have
\begin{equation}
    y_{D,\ell} = \tilde{A}_{D,\ell}x_{\ell} + \tilde{\epsilon}_D.
\end{equation}

With the randomness in $\tilde{A}_{D,\ell}$ introduced by the channels between the devices, which are assumed to be i.i.d (sub--) Gaussian, some conventional and classic CS decoding algorithms can be applied to perform the in-block support recovery. The greedy algorithms such as OMP \cite{TG07} or IHT \cite{BD09} are favorable mainly due to their low complexity, which is particularly attractive to M2M based applications where the computational capability as well as the energy constraint of the M2M devices are in general quite limited. Therefore, convex methods \cite{CR07}\cite{CT04} are mainly prohibited in such applications due to complexity reasons. \emph{Approximate Message Passing} (AMP) algorithms \cite{DMM09}, which progressively take advantage of the statistical properties to improve the convergence rate and the predictability of the algorithm, are also not preferable to large dynamic networks since an inappropriately tuned threshold function may lead to an sharply deteriorated performance \cite{BM11}. Therefore, we stick to the greedy algorithms used for the in-block support recovery at the M2M devices, and the OMP algorithm will be illustrated as an example in the following.

Moreover, we can further optimize the greedy algorithms by exploiting the feedback information from the BS on the number of active devices $|\mathcal{S}_{I,\ell}|$ in cluster $\ell$. The number of iterations needed for implementing the greedy algorithms can be limited to $|\mathcal{S}_{I,\ell}|$ since the cardinality of the support is already known, thus leading to further reduced computational complexity.
An example of the modification on OMP for in-block support recovery is summarized in Table \ref{alg1}, where the innovative steps are marked in boldface.

\begin{algorithm}
\caption{Modified OMP for In-block Support Recovery}
\label{alg1}
\begin{algorithmic}[1]
\REQUIRE $\tilde{A}$, $y_D$, $D_{\ell}$, $Q_{\ell}$, and $|\mathcal{S}_{I,\ell}|$.
\ENSURE $\mathcal{S}_{I,\ell}$. \\ \vspace{6pt}
\STATE \textbf{Form an index set $U_{\ell}=D_{\ell}\cap\bar{Q}_{\ell}$.}
\STATE \textbf{Construct a subvector $y_{D,\ell}$ of $y_D$ composed of the entries corresponding to $U_{\ell}$, and a submatrix $\tilde{A}_{D,\ell}$ of $\tilde{A}$ with vertical concatenation of rows corresponding to $U_{\ell}$ and columns corresponding to blcok $\ell$.}
\STATE Initialize the residual $r_0=y_{D,\ell}$, the index set $\Lambda_0=\emptyset$, the matrix of the chosen atoms $\Phi_0=\emptyset$, and the iteration counter $k=1$.
\STATE Choose the column of $\tilde{A}_{D,\ell}$ with index $\lambda_k$ that is best matched to $r_{k-1}$ according to
       \begin{equation}
         \lambda_k = \arg
         \max_{a_{\lambda}\in\tilde{A}_{D,\ell}}|\langle r_{k-1},a_{\lambda}\rangle|.
         \label{equ_8}
       \end{equation}
\STATE Augment the index set $\Lambda_k=\Lambda_{k-1}\cup\{\lambda_k\}$ and the matrix of the chosen atoms $\Phi_k=\left[\Phi_{k-1}\quad a_{\lambda_k}\right]$.
\STATE Solve the least square error minimization problem to obtain a new signal estimate:
      \begin{equation}
			    x_{\ell,k}=\arg \min_{x_{\ell}} ||y_{D,\ell}-\Phi_kx_{\ell}||_2. 
			\label{equ_9}
			\end{equation}
\STATE Update the residual as $r_k=y_{D,\ell}-\Phi_kx_{\ell,k}$.
\STATE \textbf{Increment $k$ by 1, and return to Step 4 until $k>|\mathcal{S}_{I,\ell}|$.}
\RETURN $\mathcal{S}_{I,\ell}=\Lambda_k$.
\end{algorithmic}
\end{algorithm}

To this end, \emph{Challenge 2} is also delicately resolved, where the proof for an accurate estimate of the in-block support $\mathcal{S}_{I,\ell}$ to be guaranteed for the M2M devices in cluster $\ell$ will be given in Section \ref{sec:theory_device}. Thus, the activation pattern $x_{\ell}$ of the $\ell$-th cluster can be precisely reconstructed and detected either by the active devices in the cluster (under full-duplex assumption) or by the cluster heads (under half-duplex assumption). Then the devices are able to discover their ranking among all the active devices in the cluster and access the corresponding assigned resource blocks by the BS based on their detected ranking.

\section{Theoretical Performance Analysis}
\label{sec:theory}

This section presents a rigorous proof of Theorem \ref{theorem_1}, which provides a sufficient condition for performance guarantees by the proposed distributed device detection scheme. Said statements summarizes two results, namely (i) the original signal can be reliably reconstructed by the proposed scheme with $\mathcal{O}(K_B\log N, K_BK_I\log d)$ measurements, and (2) the computational complexity is of $\mathcal{O}(N(K_I^2+\log N))$. The two statements will be proved and verified considering both the block support recovery at the BS and in-block support recovery at the M2M devices respectively in subsequent subsections. 

\subsection{Block Support Recovery at BS}
\label{sec:theory_BS}
\begin{lemma}
\label{lemma:theory1}
    Suppose that $x\in\mathbb{B}^N$ is an arbitrary block sparse signal with block sparsity $K_B$ and in-block sparsity $K_I$ over block size $d$, and $A\in\mathbb{R}^{M\times N}$ is the measurement matrix randomly drawn from the distribution $\mathcal{A}(R,T,L,d,\alpha)$. Given the observations $y\in\mathbb{R}^M$ in (\ref{equ_y}) and for a particular entry $x_i$ of $x$ which belongs to the $\ell$-th block and a given $t\in\{1,\cdots,T\}$, the probability that its corresponding estimate $\tilde{x}_{t,i}$ in (\ref{eq:theory_x0}) to be well estimated to the cardinality of its in-block support $|\mathcal{S}_{I,\ell}|$ with overwhelmingly small variance $\gamma$ asymptotically to 0 is given by 
		\begin{equation}
		    \mathbb{P}\left(|\tilde{x}_{t,i}-\alpha^2|\mathcal{S}_{I,\ell}||\leq 3\gamma\right)\geq 1-\frac{K_B-1}{R}.
		\end{equation}
\end{lemma}

\begin{proof}
According to (\ref{eq:theory_x1}), for a particular estimate $\tilde{x}_{t,i}$ of $x_i$ in block $\ell$ with $t\in\{1,\cdots,T\}$, the condition $|\tilde{x}_{t,i}-\alpha^2|\mathcal{S}_{I,\ell}||\leq3\gamma$ holds with overwhelming probability if the corresponding interference factor $\Delta_{t,\ell}=0$, since the noise factor $\epsilon_{t,\ell}$ is randomly drawn from a Gaussian ensemble with zero mean and variance $\gamma^2$. To this end, a sufficient (but not necessary) condition for $\Delta_{t,\ell}=0$ to hold is that the set $\mathcal{S}_B\cap\mathcal{S}_{t,\ell}\backslash\{\ell\}=\emptyset$, where $\mathcal{S}_{t,\ell}=\{j=\{1,\cdots,L\}:q_{t,j}=q_{t,\ell}\}$. This implies that $q_{t,\ell}$ is distinct from $q_{t,\bar{\ell}}$ for all $\bar{\ell}\in\mathcal{S}_B\backslash\{\ell\}$. Therefore, we have 
\begin{equation}
  \begin{aligned}
    \mathbb{P}\left(|\tilde{x}_{t,i}-\alpha^2|\mathcal{S}_{I,\ell}||\leq 3\gamma\right) &= \mathbb{P}(\Delta_{t,\ell}=0)\\
		&\geq \mathbb{P}(\mathcal{S}_B\cap\mathcal{S}_{t,\ell}\backslash\{\ell\}=\emptyset)\\
		&=\mathbb{P}\left(\forall_{\bar{\ell}\in\mathcal{S}_B\backslash\{\ell\}}\quad q_{t,\ell}\neq q_{t,\bar{\ell}}\right)\\
		&\overset{(a)}{=} \left[\mathbb{P}(q_{t,\ell}\neq q_{t,\bar{\ell}})\right]^{K_B-1}\\
		&=\left[1-\mathbb{P}(q_{t,\ell}= q_{t,\bar{\ell}})\right]^{K_B-1}\\
		&=\left(1-\frac{1}{R}\right)^{K_B-1}\\
		&\overset{(b)}{\geq} 1-\frac{1}{R}(K_B-1),
	\end{aligned}
\label{eq:theory_x2}
\end{equation}
where (a) follows since the index of rows with non-zero entries $q_{t,\ell}$ are drawn i.i.d uniformly at random for each $\ell\in\{1,\cdots,L\}$ and $|\mathcal{S}_B|=K_B$, and the inequality in (b) follows from the Bernoulli's inequality \cite{Taylor52}.
\end{proof}

\begin{lemma}
\label{lemma:theory2}
    Suppose that $x\in\mathbb{B}^N$ is an arbitrary block sparse signal with block sparsity $K_B$ and in-block sparsity $K_I$ over block size $d$, and $A\in\mathbb{R}^{M\times N}$ is the measurement matrix randomly drawn from the distribution $\mathcal{A}(R,T,L,d,\alpha)$. Given the observations $y\in\mathbb{R}^M$ in (\ref{equ_y}) and for a particular block $\ell\in\{1,\cdots,L\}$, the block estimate $\bar{x}_{\ell}$ to be well estimated to the cardinality of its corresponding in-block support $|\mathcal{S}_{I,\ell}|$ with overwhelmingly small variance $\gamma$ asymptotically to 0 is of probability
		\begin{equation}
		    \mathbb{P}\left(|\bar{x}_{\ell}-\alpha^2|\mathcal{S}_{I,\ell}||\leq 3\gamma\right)\geq 1-\frac{\delta}{L},
		\end{equation}
		if $R=\mathcal{O}(K_B)$ and $T=\mathcal{O}(\log N)$, where $0<\delta<1$ is the error tolerance.
\end{lemma}

\begin{proof}
For a given $\ell\in\{1,\cdots,L\}$, since the block estimate $\bar{x}_{\ell}$ is obtained by taking the median of $\tilde{x}_{t,i}$ over $\mathcal{O}(Td)$ samples as in (\ref{eq:theory_x3}), $\bar{x}_{\ell}$ drops within the error bound claimed by the lemma with overwhelming probability if the condition for Lemma \ref{lemma:theory1} satisfies for at least $\frac{Td}{2}$ indices of the estimates $\tilde{x}_{t,i}$.

Let $X_1,\cdots,X_{Td}$ be a set of independent (0,1) Bernoulli random variables, with each entry $X_t$ indicating whether the corresponding estimate $\tilde{x}_{t,i}$ of $x_i$ satisfies the condition for Lemma \ref{lemma:theory1}, which has probability $p\geq 1-\frac{K_B-1}{R}$ of being equal to 1. Then the probability of the number of simultaneous occurrence of the events $\{X_t=1\}$ exceeding $Td/2$ is given by
\begin{equation}
    \mathbb{P}\left(\sum_{t=1}^{Td}X_t>\frac{Td}{2}\right)=\sum_{t=\left\lfloor \frac{Td}{2} \right\rfloor +1}^{Td} \left(
		\begin{array}{c} 
	  n\\t  
	  \end{array} 
		\right)p^t(1-p)^{Td-t}.
\label{eq_chernoff_1}
\end{equation}

A lower bound on this probability can be calculated based on the Chernoff's inequality \cite{Chernoff52}, which is given as
\begin{equation}
    \mathbb{P}\left(\sum_{t=1}^{Td}X_t>\frac{Td}{2}\right)=\mathbb{P}\left(\frac{1}{Td}\sum_{t=1}^{Td}X_t>\frac{1}{2}\right)\geq 1-e^{-\frac{Td}{2p}(p-\frac{1}{2})^2},
\label{eq_chernoff}
\end{equation}

Let us take a further insight into (\ref{eq_chernoff}). The minimum probability can be easily found as achieved by $p=1/2$. Therefore, we set a lower threshold $\theta$ to the probability $p$ which lies within the interval $(\frac{1}{2}, 1]$, such that $1-\frac{K_B-1}{R}\geq \theta$. As a result, it leads to a requirement on $R$ to satisfy the condition 
\begin{equation}
    R\geq\frac{1}{1-\theta}(K_B-1)=\mathcal{O}(K_B).
\label{eq_R}
\end{equation}

In addition, if we drop the last inequality (b) in (\ref{eq:theory_x2}) and keep the probability as $p=\mathbb{P}(|\tilde{x}_{t,i}-\alpha^2|\mathcal{S}_{I,\ell}||\leq 3\gamma)\geq\left(1-\frac{1}{R}\right)^{K_B-1}$. Since we set a lower threshold $\theta$ to $p$, therefore
\begin{equation}
    \left(1-\frac{1}{R}\right)^{K_B-1} \geq \theta.
\end{equation}

Then we could obtain
\begin{equation}
		R \geq \frac{1}{1-\theta^{\frac{1}{K_B-1}}} \approx \frac{1}{-\frac{1}{K_B-1}\ln\theta}=\frac{K_B-1}{-\ln\theta}=\mathcal{O}(K_B),
\end{equation}
where the approximation follows from the limits of exponential functions since $\frac{1}{K_B-1}$ is relatively small and close to 0. The analysis results in the same requirement on $R$ as in (\ref{eq_R}).

Furthermore, (\ref{eq_chernoff}) also implies that the lower bound of the probability scales as $1-e^{-\mathcal{O}(Td)}$ for the number of indices of $\tilde{x}_{t,i}$ satisfying the conditions of Lemma \ref{lemma:theory1} exceeding $Td/2$. By taking $T=\log \left(\frac{Ld}{\delta}\right)=\mathcal{O}(\log N)$ into (\ref{eq_chernoff}), Lemma \ref{lemma:theory2} follows.
\end{proof}

\begin{lemma}
\label{lemma:theory3}
    Suppose that $x\in\mathbb{B}^N$ is an arbitrary block sparse signal with block sparsity $K_B$ and in-block sparsity $K_I$ over block size $d$, and $A\in\mathbb{R}^{M\times N}$ is the measurement matrix randomly drawn from the distribution $\mathcal{A}(R,T,L,d,\alpha)$. Given the observations $y\in\mathbb{R}^M$ in (\ref{equ_y}), the probability for the block estimate $\bar{x}_{\ell}$ to be well estimated to the cardinality of the corresponding in-block support $|\mathcal{S}_{I,\ell}|$ for all blocks $\ell\in\{1,\cdots,L\}$ with overwhelmingly small variance $\gamma$ asymptotically to 0 is given by
		\begin{equation}
		    \mathbb{P}\left(\forall_{\ell\in\{1,\cdots,L\}}:|\bar{x}_{\ell}-\alpha^2|\mathcal{S}_{I,\ell}||\leq 3\gamma\right)\geq 1-\delta.
		\end{equation}
		if $R=\mathcal{O}(K_B)$ and $T=\mathcal{O}(\log N)$, where $0<\delta<1$ is the error tolerance.
\end{lemma}

\begin{proof}
The condition claimed by the lemma satisfies \emph{if and only if} (iff) the condition for Lemma \ref{lemma:theory2} holds for all $\ell\in\{1,\cdots,L\}$. Therefore, we have
\begin{equation}
  \begin{aligned}
    \mathbb{P}\left(\forall_{\ell\in\{1,\cdots,L\}}:|\bar{x}_{\ell}-\alpha^2|\mathcal{S}_{I,\ell}||\leq 3\gamma\right)&=\left(\mathbb{P}\left(|\bar{x}_{\ell}-\alpha^2|\mathcal{S}_{I,\ell}||\leq 3\gamma\right)\right)^L\\
		&\geq\left(1-\frac{\delta}{L}\right)^L\\
		&\geq1-L\cdot\frac{\delta}{L}\\
		&=1-\delta,
	\end{aligned}
\label{eq_lemma2}
\end{equation}
where the first equation follows since the estimate for each block is
i.i.d, and the following inequality follows from Bernoulli's inequality \cite{Taylor52} since $\delta/L\ll1$.
\end{proof}

\begin{corollary}
\label{lemma:corollary1}
    The computational complexity for reliable block support recovery at the BS is of $\mathcal{O}(N\log N)$.
\end{corollary}

\begin{proof}
According to (\ref{eq:theory_x0}), for a particular entry $x_i$ of $x$ which belongs to the $\ell$-th block and a given $t\in\{1,\cdots,T\}$, the calculation of its corresponding estimate $\tilde{x}_{t,i}$ only involves 1 multiplication. Thereafter, the block-wise estimate $\bar{x}_{\ell}$ in (\ref{eq:theory_x3}) for block $\ell\in\{1,\cdots,L\}$ is obtained as the median of $\tilde{x}_{t,i}$ over $\mathcal{O}(Td)$ samples. The computational complexity for finding the median of an unsorted array with $N$ elements is of $\mathcal{O}(N)$ by using the median-of-medians algorithm \cite{LCSR09}. Therefore, the calculation of the block-wise estimate $\bar{x}_{\ell}$ requires $Td$ times multiplication and $\mathcal{O}(Td)$ operations to find the median, thus the computational complexity remains as $\mathcal{O}(Td)$. Taking the same operation over all $L$ blocks, the computational complexity for block support recovery at the BS turns to be $\mathcal{O}(TdL)=\mathcal{O}(TN)$. Lemma \ref{lemma:theory3} has proved that robust performance guarantee can be achieved at the BS if $R=\mathcal{O}(K_B)$ and $T=\mathcal{O}(\log N)$. Taking the requirement on $T$ yields the overall computational complexity for reliable block support recovery at the BS to be $\mathcal{O}(N\log N)$.
\end{proof}

To this end, we have shown that by choosing $R$ and $T$ large enough, i.e., $R=\mathcal{O}(K_B)$ and $T=\mathcal{O}(\log N)$ (for a total of $M=\mathcal{O}(K_B\log N)$ measurements), reliable block support recovery at the BS can be guaranteed with overwhelming probability and with computational complexity of $\mathcal{O}(N\log N)$.

\subsection{In-block Support Recovery at Devices}
\label{sec:theory_device}
\begin{lemma}
\label{lemma:theory4}
    Suppose that $x\in\mathbb{B}^N$ is an arbitrary block sparse signal with block sparsity $K_B$ and in-block sparsity $K_I$ over block size $d$, and $A\in\mathbb{R}^{M\times N}$ is the measurement matrix randomly drawn from the distribution $\mathcal{A}(R,T,L,d,\alpha)$. Given the observations $y_D\in\mathbb{R}^M$ in (\ref{equ_yD}), the activation pattern $x_{\ell}$ for cluster $\ell\in\{1,2,\cdots,L\}$ with sparsity $K_I$ and of dimension $d$ can be reliably recovered with the modified OMP algorithm if $R=\mathcal{O}(K_B)$ and $T=\mathcal{O}(K_I\log d)$, and the computational compleixty is of $\mathcal{O}(NK_I^2)$.
\end{lemma}

\begin{proof}
According to the specific structure of the measurement matrix $A$, for a given cluster $\ell\in\{1,\cdots,L\}$, the corresponding submatrix $A_{-,\ell}$ has only $T$ rows with non-zero components indexed by the set $D_{\ell}$. Thus, in order to perform the in-block support recovery of $x_{\ell}$ at an M2M device from cluster $\ell$ and with the feedback information from the BS on the collision patten $Q_{\ell}$ in the measurements, the effective measurements that can be used for the decoding process consists of the $T$ entries of $y_D$ corresponding to $D_{\ell}$ except for those ``collided'' measurements indexed by $Q_{\ell}$. Since Lemma \ref{lemma:theory1} provides the guarantee to obtain an effective measurement that is collision-free and each estimate is i.i.d, the overall number of effective measurements $T_I$ for the in-block support recovery of an active cluster $\ell\in\mathcal{S}_B$ can be estimated as
\begin{equation}
    T_I=T\cdot\mathbb{P}\left(|\tilde{x}_{t,i}-\alpha^2|\mathcal{S}_{I,\ell}||\leq 3\gamma\right)\geq T\cdot\left(1-\frac{K_B-1}{R}\right).
\label{eq:TI}
\end{equation} 

In \cite{TG07}, it has been proved that if the measurement matrix is an i.i.d (sub--) Gaussian matrix or random Bernoulli matrix, then a $K$-sparse signal of dimension $N$ can be reliably reconstructed with greedy algorithms like OMP if the number of measurements $M=\mathcal{O}(K\log N)$. The same requirement has also been proved in \cite{CW14} and \cite{WZWTM16} for reliable decoding using OMP under noisy conditions.
For the in-block support recovery at an M2M device from cluster $\ell$, since the activation pattern $x_{\ell}$ to be reconstructed is of dimension $d$ and with sparsity level $K_I$, the support can be recovered with overwhelming probability if the number of effective measurements $T_I$ satisfies
\begin{equation}
    T_I\geq T\cdot\left(1-\frac{K_B-1}{R}\right)=\mathcal{O}(K_I\log d).
\label{eq:theory_TI}
\end{equation}

Recall that $1-\frac{K_B-1}{R}\geq\theta, \theta\in(\frac{1}{2}, 1]$ should be satisfied for the block support recovery procedure at the BS with $R=\mathcal{O}(K_B)$, thus the probability term turns into a constant value. Therefore, we have $T=\mathcal{O}(K_I\log d)$.

It is also verified in \cite{TG07} that the computational complexity is $\mathcal{O}(NK^2)$ for sparse signal recovery via OMP. Thus, since the signal to be reconstructed for in-block support recovery at an M2M device is of dimension $d$ and with sparsity level $K_I$, $\mathcal{O}(dK_I^2)$ operations are sufficient for decoding with the modified OMP algorithm. Taking the same action over all $L$ blocks, the overall computational complexity becomes $\mathcal{O}(LdK_I^2)=\mathcal{O}(NK_I^2)$. 
\end{proof}


\subsection{Proof of Theorem \ref{theorem_1}}
\label{sec:theory_condition}

For a given realization of the measurement matrix, it has been proved in Lemma \ref{lemma:theory3} that $R=\mathcal{O}(K_B)$ and $T=\mathcal{O}(\log N)$ are sufficient for reliable block support recovery at the BS, and in Lemma \ref{lemma:theory4} that $R=\mathcal{O}(K_B)$ and $T=\mathcal{O}(K_I\log d)$ measurements are required for the in-block support recovery at the M2M devices. Taking the maximum value of both sides yields the sufficient guarantee on the performance regarding the number of measurements

\begin{equation}
    M = \mathcal{O}(\max\{K_B\log N, K_BK_I\log d\}).
\end{equation}

As $M$ is the dimension of the unique signature transmitted by an active M2M device, it is also an indication of the signal acquisition time as far as distributed schemes are concerned. 
Furthermore, since the algorithms for block support recovery at the BS requires $\mathcal{O}(N\log N)$ operations and the in-block support recovery at the device side requires $\mathcal{O}(NK_I^2)$ operations, the overall computational complexity is of $\mathcal{O}(N(K_I^2+\log N))$ for the whole detection process. 

However, if the signal is treated as a conventional $K$-sparse vector (where $K=K_BK_I$) as in \cite{Don06} without exploiting knowledge of the block-sparse structure, a sufficient condition for reliable signal recovery would be $M=\mathcal{O}(K\log N)=\mathcal{O}(K_BK_I\log N)$, and the computational complexity is of $\mathcal{O}(NK^2)$. Since $d\ll N$ and $K_I\ll K$, we can see that from the scaling point of view, less signal acquisition time is required by the proposed scheme with significantly reduced computational complexity.

\section{Numerical Results}
\label{sec:simulation}
We conduct extensive simulations to verify the performance of the proposed distributed device detection and resource allocation scheme. In our experiments, we take the number of M2M devices in the network to be $N=10000$ and they are partitioned into $L=100$ clusters with equal size $d=100$. The sparsity level $K=K_BK_I$ is set within the interval [10 100]. Thus, the target problem is to reconstruct the $K$-sparse binary vector of length $N$ from $M$ distributed measurements obtained via the measurement matrix $A\in\mathbb{R}^{M\times N}$ which is randomly drawn from the distribution $\mathcal{A}(R,T,L,d,\alpha)$ as defined in Section \ref{sec:algorithm_matrix}.

We first compare the performance of the proposed scheme to CS based approaches, among which we take the classic greedy algorithm OMP with random Gaussian measurements as the baseline. We assume that the signal to be reconstructed is treated as a conventional $K$-sparse vector, and centralized decoding is performed without exploiting knowledge of the block-sparse structure. Moreover, with respect to distributed processing as proposed in this work, the number of measurements $M$ is thus an indication of the acquisition time of the transmitted signal.
The sparsity level of the signal is set to be $K=20$, with block sparsity $K_B=4$ and in-block sparsity $K_I=5$, respectively. For each plot we average over 1000 pairs of realizations of the measurement matrix and the block-sparse signal. 

Figure \ref{fig_1} depicts the detection probability as a function of the signal acquisition time (in terms of $M$) under both noisy and noise-free conditions for the proposed scheme as well as for the baseline. The performance for the noisy case is evaluated by setting the \emph{signal-to-noise ratio} (SNR) to 5 dB in the simulations.
We can see that the proposed scheme outperforms the standard OMP algorithm with Gaussian measurements under both noisy and noisy-free conditions, where less acquisition time is required by the proposed scheme to achieve the same performance as OMP. Furthermore, the performance gap between the noisy and the noise-free case by the proposed scheme is comparatively smaller compared with the baseline, indicating more robustness in performance against noise.

Figure \ref{fig_2} plots the CDF of the average number of iterations for the signal recovery by standard OMP with random Gaussian measurements and the modified OMP in Table \ref{alg1}. It can be obviously shown that the modified OMP significantly outperforms the standard OMP algorithm with random Gaussian measurements by requiring much less number of iterations, since the dimension of the signal to be reconstructed is significantly reduced to the block size and the stopping criteria for running the algorithm is sternly confined to the cardinality of the detected in-block support.
Therefore, the computational complexity is drastically reduced in view of distributed processing by the proposed scheme.

    \begin{figure}[htb]
				\centering
					\includegraphics[width=3.2in]{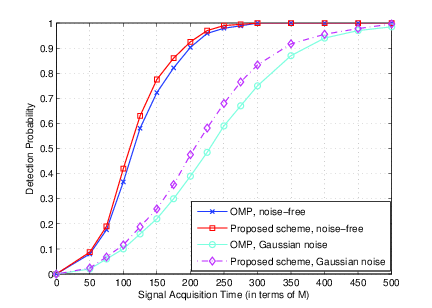}
        \caption{Performance comparison on the detection probability between the proposed scheme and standard OMP with random Gaussian measurements.}
        \label{fig_1}
    \end{figure}
		
		\begin{figure}[htb]
   		\centering
					\includegraphics[width=3.2in]{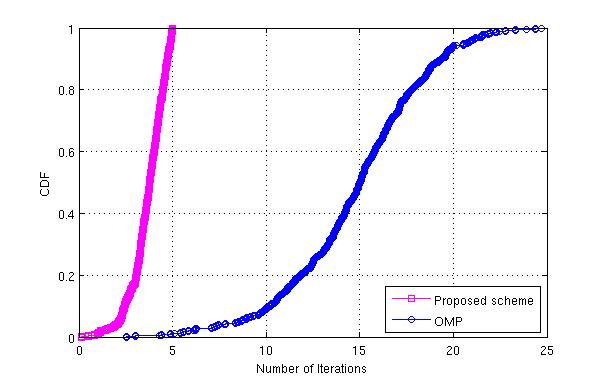}
        \caption{Comparison on the CDF of the average number of iterations for signal recovery between standard OMP with random Gaussian measurements and the modified OMP in Table \ref{alg1}.}
        \label{fig_2}
    \end{figure}

We also compare the performance of the proposed scheme with two classical access schemes, namely the LTE RA procedure \cite{3GPP321} and the conventional cluster-based approach \cite{JPL10} where a cluster head aggregates and forwards messages/requests for the rest of the devices in the cluster and initiates LTE RA procedure on behalf of the cluster members. We set the number of measurements $M$ to be 839 bit in the simulation -- the same as the length of Zadoff-Chu sequence \cite{3GPP321} used for the LTE RA procedure, thus the three schemes are running with the same signal acquisition time. The sparsity level $K=K_BK_I$ is again set within the range between 10 and 100. 

Figure \ref{fig_3} depicts the detection probability by the three schemes as a function of the number of active devices in the network, i.e., the sparsity level. It can be easily observed that the proposed scheme significantly outperforms the LTE RA procedure, thus achieving much better scalability with the increasing number of M2M devices deployed in the network and leading to more robust performance in the detection process. Moreover, the proposed scheme also achieves better performance than the cluster-based approach if the sparsity level is sufficiently large.

Figure \ref{fig_4} plots the access delay performance by the three schemes. Since the proposed scheme calls for much less coordination and signaling exchange between the M2M devices as well as to the infrastructure, the signaling overhead is substantially reduced and thus leading to significantly decreased access latency.

    \begin{figure}[htb]
				\centering
					\includegraphics[width=3.2in]{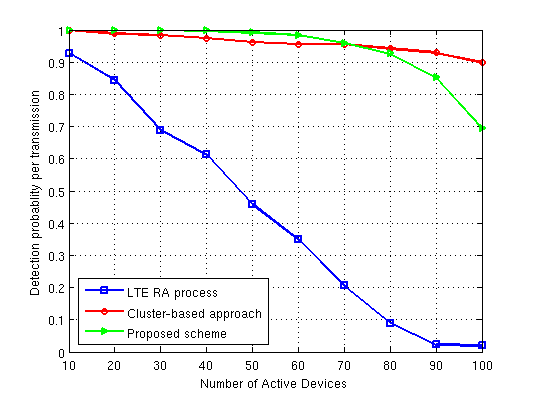}
        \caption{Performance comparison on the detection probability between the proposed scheme, LTE RA procedure and conventional cluster-based approach \cite{JPL10} with sparsity level within the interval [10 100].}
        \label{fig_3}
    \end{figure}
		
		\begin{figure}[htb]
   		\centering
					\includegraphics[width=3.2in]{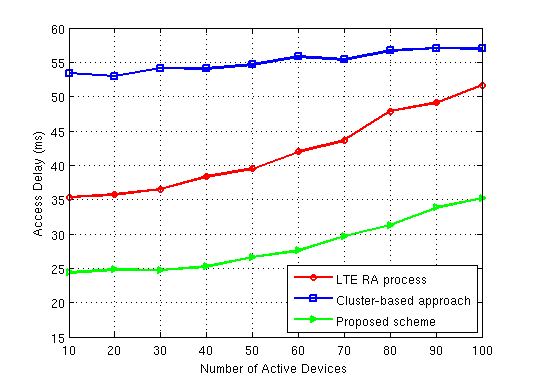}
        \caption{Performance comparison on the averaged access delay between the proposed scheme, LTE RA procedure and conventional cluster-based approach \cite{JPL10} with sparsity level within the interval [10 100].}
        \label{fig_4}
    \end{figure}

\section{Conclusion}
\label{sec:conclude}
This work utilizes the framework of CS for distributed detection of the network activation pattern to facilitate efficient resource allocation in large-scale M2M communication networks. The correlation in the device behaviors and the sparsity in the activation pattern of the M2M devices are fully exploited, thus mapping the target problem into a support recovery procedure for a particular block sparse signal -- with additional in-block structure -- in CS based applications. The detection techniques are mainly based on sketching and greedy algorithms, which inherit the virtues of low sample complexity and low computational complexity. By applying the distributed schemes and the optimization in the algorithms, it has been verified that a $K$-sparse binary vector $x\in\mathbb{B}^N$ with block sparsity $K_B$ and in-block sparsity $K_I$ over block size $d$ can be reliably reconstructed within an acquisition time of $\mathcal{O}(\max\{K_B\log N, K_BK_I\log d\})$ with computational complexity of $\mathcal{O}(N(K_I^2+\log N))$, which achieves a better scaling compared with conventional CS based approaches without exploiting the particular block-sparse structure. Furthermore, compared with conventional schemes including classic CS based approaches, the LTE RA procedure and cluster-based access mechanisms, the proposed scheme achieves better scalability with the increasing network size and sufficiently reduces the computational complexity as well as the signaling overhead, thus leading to more robust performance in the detection process, especially in terms of higher detection probability and reduced access delay.

However, there are still some limitations on the proposed approach, especially on the requirement of perfect synchronization during the acquisition phase and priori knowledge of the channel state by the M2M devices. Further extensions to relax these limitations will be investigated in future work.


%


\section*{Acknowledgment}

The authors would like to acknowledge partial support by the DFG in the context of grants JU 2795/2\&3 and STA 864/8-1.

\ifCLASSOPTIONcaptionsoff
  \newpage
\fi

\end{document}